\documentclass[runningheads,a4paper]{llncs}
\makeatletter
\pdfoutput=1
\usepackage{graphicx}
\graphicspath{./Figures/}
\usepackage{chngcntr}
\usepackage{amssymb}
\usepackage{booktabs}
\usepackage{lineno}
\usepackage{amsmath}
\usepackage{amsfonts}

\newcommand{\deleted}[1]{}
\usepackage{comment}
\usepackage{amssymb}
\usepackage{layouts}
\usepackage{xcolor}
\newcommand{\freg}{\textit{freg}}
\newcommand{\fregc}{\overline{\textit{freg}}}

\newcommand{\gm}{\mathrm{Gmap}}
\newcommand{\fvd}{\mathrm{FVD}}
\newcommand{\Gs}{G^s}

\newcommand{\T}{\mathrm{FVD}}
\newcommand{\newapptheorem}[2]{
  \newtheorem{#1}{#2}[section]
 \counterwithin{#1}{section}}
\newapptheorem{apptheorem}{Theorem}
\newapptheorem{applemma}{Lemma}
\newapptheorem{appdef}{Definition}

\begin{document}
\mainmatter

\title{Linear-Time Algorithms for the Farthest-Segment
  Voronoi Diagram and Related Tree Structures\thanks{
    Research supported in part by
    the Swiss National Science Foundation, projects
    SNF 20GG21-134355 (ESF EUROCORES EuroGIGA/VORONOI) and  SNF
    200021E-154387.
    }}
\titlerunning{Linear-Time Algorithms for FsVD  and Related Tree
  Structures}
\title{An Expected Linear-Time Algorithm for the Farthest-Segment
  Voronoi Diagram\thanks{
    Research supported in part by
    the Swiss National Science Foundation, projects
    SNF 20GG21-134355 (ESF EUROCORES EuroGIGA/VORONOI) and  SNF 200021E-154387.
    }}
\titlerunning{Expected Linear-Time Algorithm for FsVD}

\author{Elena Khramtcova\and
        Evanthia Papadopoulou}

\institute{Faculty of Informatics,  Universit\`a della Svizzera
  italiana (USI), Lugano, Switzerland\\
 {\tt \{elena.khramtcova,  evanthia.papadopoulou\}@usi.ch}}

\maketitle

\begin{abstract}
We present an expected linear-time algorithm to construct 
the farthest-segment Voronoi diagram, given the sequence of its faces
at infinity.
This sequence forms a
Davenport-Schinzel sequence of order $3$ and it
 can be computed in $O(n\log n)$ time, where $n$ is the number
of input segments.
The farthest-segment Voronoi diagram is a tree, 
with disconnected Voronoi regions, of total  complexity $\Theta(n)$ in
the worst case.
Disconnected regions pose a major difficulty in
 deriving linear-time construction  algorithms for 
such tree-like Voronoi diagrams. 
In this paper we present a new approach towards this direction.
\end{abstract}

\section{Introduction}
It is well known that the Voronoi diagram of points in convex position 
can be computed in linear time, given 
the order of their convex hull~\cite{AGSS89}.
This linear-time construction extends to a class of related  diagrams
such as the farthest-point Voronoi diagram given their convex hull, the medial
axis of a convex polygon, and updating a nearest-neighbor Voronoi
diagram of points after deletion of one site.
In an abstract setting, 
a \emph{Hamiltonian abstract Voronoi diagram} 
can be computed in linear time~\cite{KL94}, given the order
of Voronoi regions along an 
unbounded simple curve, which visits each
region exactly once and can  intersect each bisector only once.
This construction has been extended recently to  include 
forest structures~\cite{BKLL18} under 
similar conditions where no region can have multiple faces within the domain enclosed by the curve.
The medial axis of a simple polygon can also be computed
in linear time~\cite{snoyeink99}.
It is therefore natural to ask what other types of Voronoi 
diagrams 
can be constructed in linear time.
In this paper we consider the farthest Voronoi diagram of line segments.

Classical variants of Voronoi diagrams such as  higher-order 
Voronoi diagrams for sites other than points, 
had been surprizingly ignored in the
  literature of computational geometry until recently \cite{bcklpz15,PZ16}. 
Given a set  $S$ of $n$ simple geometric objects in the plane, called sites, 
the \emph{order-$k$ Voronoi diagram} of $S$ is a partitioning of the
plane into regions such that every point within a 
region has the same $k$ nearest sites.  
For $k=1$, this is the \emph{nearest-neighbor Voronoi diagram} and
for $k=n-1$ it is the \emph{farthest-site  Voronoi diagram} of $S$.
Despite similarities, these diagrams for 
non-point sites, e.g., line segments,
can  illustrate
fundamental structural differences from their counterparts for points, 
such as the presence of disconnected regions (see also
\cite{AuDrysdale06,fpvd11,AbstractFarthestVoronoi}).
However, segment Voronoi diagrams are 
fundamental to problems involving proximity among polygonal objects,
see, e.g.,~\cite{P11} and references therein for 
 applications  in VLSI design.
For more information on Voronoi diagrams see the book of
Aurenhammer et al.~\cite{Aurenbook}.

In this paper we present an expected linear-time algorithm to construct the farthest segment Voronoi
diagram, after the sequence of its  faces at infinity is known.
Our approach is inspired by the
randomized approach of Chew~\cite{Chew90} to construct the Voronoi diagram
of points in convex position.
It establishes the means to deal with disconnected regions
within this  framework 
and we expect it to be applicable in other cases of tree-like diagrams.
A major difference with  respective problems for points is that the
sequence of Voronoi faces along a relevant  enclosing boundary
forms a Davenport-Schinzel sequence of order $\geq 2$ (order $3$ for
the farthest-segment Voronoi diagram) 
in contrast to the case of points, where no repetition can exist.
Repetition introduces several complications, including the fact that 
the sequence of Voronoi faces  at infinity 
for a subset of
the original segments, $S'\subset S$,
is not a subsequence  of the  respective sequence for $S$.
In addition, such a subsequence may not even correspond to a Voronoi
diagram.
The intermediate diagrams 
computed by our algorithm are interesting in their
own right. They have the structural properties of 
segment Voronoi diagrams, however, 
they do not correspond to a segment diagram, nor are they instances of
abstract Voronoi diagrams. 

The resulting algorithm in this paper is very simple, no more
complicated than its counterpart for points (see e.g., \cite{CGbook}).
The major challenge is formulating  and handling  the 
intermediate structures. We expect these structures to also be useful in
considering the linear-time framework of Aggarwal et
al.~\cite{AGSS89}, which remains an open problem.%
\footnote{
  A preliminary version
  of this  paper \cite{ISAAC15} contains a gap in adapting the
  framework of \cite{AGSS89}, thus, a deterministic linear-time algorithm 
  remains open.}

This paper is organized as follows. Section~\ref{sec:def}
gives all necessary preliminaries and definitions, including the
definition of an arc sequence of the \emph{Gaussian map} of a set of line
segments (see Section~\ref{sec:subsequences}). Section~\ref{sec:fvd}
introduces the farthest Voronoi diagram of an arc sequence and studies
its properties. In Section~\ref{sec:del-ins} we introduce the
operations of insertion and deletion of an arc from an arc sequence
and its farthest Voronoi diagram. Section~\ref{sec:randomized} gives
an expected linear-time algorithm to compute 
the farthest Voronoi diagram for certain arc sequences, which implies an algorithm to compute
the farthest-segment Voronoi diagram. Finally, we give brief concluding remarks.

\section{Preliminaries and Definitions} 
\label{sec:def}
Let $S$ be a set of arbitrary line segments in $\mathbb{R}^2$;
segments in $S$  may
intersect or touch at  a single point.
The distance between a point $q$ and a line segment $s_i$ is
$d(q,s_i)=\min\{d(q,y) \mid y\in s_i\}$, where $d(q,y)$ denotes
the ordinary distance between two points $q,y$
 in  the $L_2$ (or the $L_p$) metric.
In this paper we focus on the $L_2$ metric, however,
  our technique applies also to $L_p$, $p \geq 1$.

\begin{figure}
\begin{minipage}{0.49\linewidth}
\centering
\includegraphics[page=1]{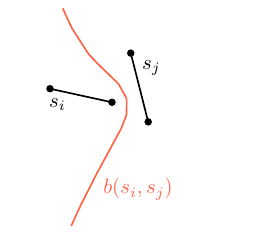} \\ 
(a)
\end{minipage}
\hfill
\begin{minipage}{0.49\linewidth}
\centering
\includegraphics[page=2]{2d-bis} \\ 
(b)
\end{minipage}
\caption{The bisector of two  segments  $s_i,s_j$. Two cases: (a) $s_i$ and $s_j$ are disjoint, 
and (b) $s_i$ and $s_j$ share an endpoint.}
\label{fig:bisector}
\end{figure}

The bisector of two segments  $s_i,s_j\in S$ is 
$b(s_i,s_j) =  \{x \in \mathbb{R}^{2} \mid d(x, s_{i}) = d(x,s_{j})\}. $ 
For disjoint segments, $b(s_i,s_j)$ is an unbounded  curve that consists of a constant
number of  pieces, where each piece is a portion of an elementary bisector between the
endpoints and open portions of $s_i, s_j$, see Fig.~\ref{fig:bisector}a.
If two segments intersect at point $p$ (other than their  common endpoint)
their bisector consists of two such curves intersecting at $p$.
The unbounded pieces of $b(s_i,s_j)$ are
rays of 
bisectors between the segment endpoints oriented towards infinity. 
We refer to the direction of such a ray
as a \emph{direction of $b(s_i,s_j)$}.
Bisector $b(s_i,s_j)$ has two directions, if $s_i, s_j$ are disjoint, and four if $s_i, s_j$ intersect transversally. 
If  segments $s_i, s_j$
have a common endpoint then $b(s_i,s_j)$ contains a 2-dimensional
region, see the shaded region in Fig.~\ref{fig:bisector}b. Following  standard conventions, see e.g.~\cite{AuDrysdale06},
the 2-dimensional 
region can be replaced by the piece of
the angular bisector of $s_i,
s_j$ within this 
region, obtaining a single curve, see the red curve in Fig.~\ref{fig:bisector}b.

The farthest Voronoi region of a  segment $s_{i}$ is 
$\textit{freg}(s_{i}) = \{x \in \mathbb{R}^{2} \mid d(x, s_{i}) > 
 d(x,s_{j}), 1 \leq j \leq n,  j\neq i \}$. 
The (non-empty) farthest 
Voronoi  regions of
the segments in $S$, together with their bounding edges and vertices, 
define a partition of the plane, called the \emph{farthest-segment 
Voronoi diagram} of $S$, denoted $\fvd(S)$,
see Fig.~\ref{fig:fvd-hull-gmap}a.
Any maximally connected subset of a Voronoi region is called a \emph{face}.

\begin{figure}
\begin{center}
\includegraphics{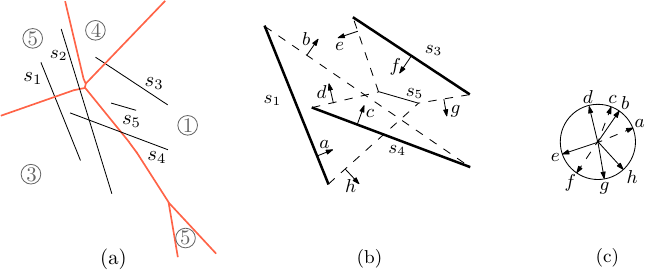}
\caption{\cite{PD13} (a) $\fvd(S)$,  $S = \{s_1,\dots,s_5\}$; (b) its farthest hull; (c) $\gm(S)$}
\label{fig:fvd-hull-gmap}
\end{center}
\end{figure}

In the following we review the notions of a \emph{hull} and its \emph{Gaussian
  map}  that characterize the faces of the farthest-segment Voronoi
diagram and the directions along which they are unbounded.
We also define an \emph{arc sequence} that is used
throughout this paper.

A farthest Voronoi region $\freg(s_i)$  is non-empty 
and unbounded in direction $\phi$ if and only if there
exists an open halfplane, normal to $\phi$, which intersects all
segments in $S$ but $s_i$~\cite{AuDrysdale06}. 
The line  $\ell$, normal to $\phi$,
bounding such a halfplane, is called a \emph{supporting line}.
The direction $\phi$ (normal to $\ell$) is called
the \emph{hull direction} of $\ell$ and is denoted $\nu(\ell)$. 
An unbounded Voronoi edge separating
$\freg(s_i)$ and $\freg(s_j)$ contains  a portion of $b(p,q)$, for
endpoints $p,q$ of $s_i,s_j$ respectively, where 
the line through  $\overline{pq}$  induces  an open halfplane 
 that intersects all segments in 
$S$, except  $s_i,s_j$  (and  possibly except additional  segments incident to
$p,q$). Segment  $\overline{pq}$ is called a \emph{supporting segment};
the direction normal to it pointing to the inside
of this halfplane is called its
\emph{hull direction}  and it is denoted $\nu(\overline{pq})$. 
A segment $s_i\in S$ such that the line $\ell$ through $s_i$ is supporting,
is called a \emph{hull segment}; 
 its \emph{hull direction} $\nu(s_i)$ is $\nu(\ell)$.
The closed polygonal curve obtained by following the supporting and
hull segments in the angular order of their hull directions is called the
\emph{farthest hull}~\cite{PD13}.
Figs.~\ref{fig:fvd-hull-gmap}a and~\ref{fig:fvd-hull-gmap}b 
illustrate a farthest-segment Voronoi
diagram and its hull respectively. In Fig.~\ref{fig:fvd-hull-gmap}b, supporting segments are shown in 
dashed lines; hull segments are shown in bold; 
arrows indicate the hull directions of all
supporting and hull segments.

The Gaussian map of $\fvd(S)$, denoted $\gm(S)$  (see
Fig.~\ref{fig:fvd-hull-gmap}c), provides a correspondence between
the faces of  $\fvd(S)$ and a \emph{circle of directions} $K$~\cite{PD13}.
 $K$ can be assumed to be a unit circle, where
each point $x$ on $K$ corresponds to a direction as indicated by the
radius of $K$  at $x$. Each Voronoi face is mapped to an arc on $K$,  which represents  the set
of directions along which the face is unbounded. 
The  $\gm(S)$ can be viewed as a cyclic sequence of arcs on $K$ 
(in counterclockwise order), where each arc corresponds to one face of  $\fvd(S)$. 
Two neighboring arcs  $\alpha, \gamma$ are separated by $\nu(\alpha,
\gamma)=\nu(\overline{pq})$, where $p,q$ are the endpoints of
$s_\alpha,s_\gamma$ forming a supporting segment;
$\nu(\alpha, \gamma)$ is the relevant direction of bisector $b(p,q)$. 
Fig.~\ref{fig:fvd-hull-gmap}c shows $\gm(S)$ for the segment set $S$ of 
Fig.~\ref{fig:fvd-hull-gmap}a; 
hull directions for hull and supporting segments 
are shown as dashed and solid arrows respectively.
The arc on $\gm(S)$ of a hull segment $s$ is called a \emph{segment arc};
it consists of two sub-arcs, one for each endpoint of $s$,  separated by 
$\nu(s)$ (see e.g., the dashed vectors $a,c,f$ in
Fig.~\ref{fig:fvd-hull-gmap}c). 
An arc that corresponds to a single
endpoint of a segment is called a \emph{single-vertex} arc.
The  $\gm(S)$ can be computed in $O(n\log n)$  time (or output-sensitive $O(n\log
h)$ time, where $h=|\gm(S)|$)~\cite{PD13}.

The  standard point-line duality transformation $T$ offers 
another  correspondence between the faces of $\fvd(S)$ and
envelopes of \emph{wedges}~\cite{AuDrysdale06}. 
It maps a point $p=(a,b)$ in the primal
plane to a line $T(p): y=ax-b$ in the dual plane, and vice versa.
A segment $s_i=uv$ corresponds  to a \emph{lower wedge}
and to an \emph{upper wedge}, 
defined  respectively 
as the area below and the area above 
both lines $T(u)$
and $T(v)$.
Let  $E$ (resp., $E'$)
be the boundary of the union of the lower (resp., upper) wedges. 
The faces of  $\fvd(S)$ correspond exactly to the edges of $E$ and $E'$~\cite{AuDrysdale06}. 

Let the upper and lower Gmap be the portion of $\gm(S)$ above and
below respectively the horizontal diameter of $K$.
There is a clear correspondence between $E$ (resp., $E'$) 
and the upper (resp., lower)
Gmap: the vertices of $E$
are exactly the  hull directions of supporting segments on the upper
Gmap (corresponding to directions of bisectors)  and
the apexes of wedges in $E$ are exactly the hull directions of hull
segments~\cite{PD13}; see Fig.~\ref{fig:arc-seq}a and \ref{fig:arc-seq}b.
For both $E$ and $E'$,  their left-to-right order of
increasing x-coordinate corresponds to a counterclockwise order of
directions in $\gm$.

\subsection{Arc sequences and subsequences of $\gm(S)$}
\label{sec:subsequences}

Consider the two arrangements of upper and lower wedges for 
the segments in $S$.
Any $x$-monotone path $\pi$ in the arrangement of lower
(resp., upper) wedges can be
transformed  into a sequence of arcs in the 
upper (resp., lower) half-circle of $K$. 
Fig.~\ref{fig:arc-seq}b shows two such  paths in blue and red lines
respectively.
The left-to-right traversal of path $\pi$ corresponds to a counterclockwise traversal of the arc sequence. 
The arc sequence corresponding to the red path  is shown in Fig.~\ref{fig:arc-seq}c.

\begin{figure} 
\begin{center}
\includegraphics{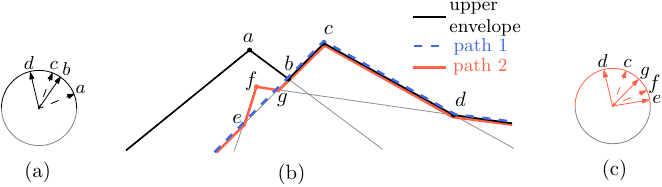}
\caption{(a) The upper Gmap of $S$ from Fig.~\ref{fig:fvd-hull-gmap}; (b) the dual arrangement of lower wedges with its upper envelope (black) and two other $x$-monotone 
paths, shown  in blue (dashed) and in red. 
(c) upper Gmap corresponding to the red path.} 
\label{fig:arc-seq}
\end{center}
\end{figure}

\begin{definition}
A sequence of arcs on the circle of directions $K$ 
that corresponds to a pair of $x$-monotone paths,
one in the arrangement of upper and one in the arrangement of lower
wedges, is called an \emph{arc sequence}.
\end{definition}

An arc sequence $G$ may contain consecutive arcs of the same
segment.
The maximal union of such consecutive arcs is called a \emph{maximal
  arc}. 
The \emph{simplified version} of $G$, denoted $\Gs$, consists of unified
maximal arcs for any consecutive arcs of the same
segment in $G$.
In a simplified arc sequence, any two neighboring arcs 
belong to two different segments.
In an arc sequence, no arcs  can overlap and no gaps can be present.

An arc sequence $G$ is  called a \emph{subsequence} of $\gm(S)$ if
every arc of $G$ entirely contains an arc of $\gm(S)$ that is induced
by the same segment. The arcs in $G$ are simply expanded versions of
the arcs in $\gm(S)$. 
Fig.~\ref{fig:arc-seq}b illustrates in dual space a subsequence of $\gm(S)$ 
in dashed blue lines (path 1).

The arcs in $\gm(S)$ as well as their expanded
versions in any subsequence $G$ of $\gm(S)$ are called \emph{original arcs}.
The exact arcs of $\gm(S)$ (with fixed endpoints) are called \emph{core arcs}. An original
arc always contains its corresponding core arc, 
which is referred to as its \emph{core}.
The core of an arc $\beta\in G$ is denoted as $\beta^*$.

An arc  sequence $G'$  is called an \emph{augmented subsequence} of
$\gm(S)$, if in addition to \emph{original} arcs of $\gm(S)$, 
it also contains some \emph{new
arcs} such that for every new arc in $G'$ there is an original arc of
the same segment in $G'$.
Fig.~\ref{fig:arc-seq}b illustrates in dual space an augmented
subsequence $G'$ of $\gm(S)$ in red solid lines (path 2).
The parts of the augmented subsequence before point $g$,  correspond to new arcs, 
and the parts after point $g$ are original arcs. 
An augmented subsequence  $G'$, which has the same original arcs as
a subsequence  $G$, is said to be \emph{corresponding} to $G$.
In dual space, $G'$ is an $x$-monotone path lying between the path of
$G$ and the envelope of the arrangement (see the red and blue paths
respectively in Fig.~\ref{fig:arc-seq}b).

Throughout this paper, given an arc $\alpha$, let $s_\alpha$ denote
the segment in $S$ that induces $\alpha$.

\section{The Farthest Voronoi Diagram of an Arc Sequence} 
\label{sec:fvd}

In this section we define the farthest Voronoi diagram of an arc
sequence $G$,  denoted $\fvd(G)$, 
where $G$ is a (possibly augmented) subsequence of  $\gm(S)$.
Such diagrams appear as intermediate structures in the process of
computing  $\fvd(S)$. 
If $G=\gm(S)$, then $\fvd(G)=\fvd(S)$, see Fig.~\ref{fig:base-case}.

\begin{figure}[h]
\begin{minipage}{0.49\linewidth}
\centering
\includegraphics[page = 1]{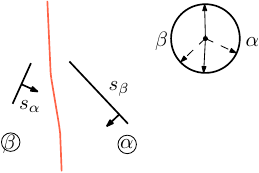}\\
(a)
\end{minipage}
\hfill
\begin{minipage}{0.49\linewidth}
\centering
\includegraphics {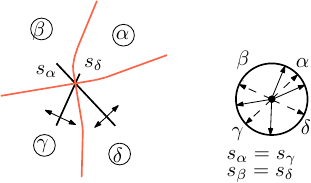}\\
(b)
\end{minipage}
\caption{$\fvd(S)=\fvd(G)$, where $G = \gm(S)$ for  
(a) set $S = \{s_\alpha,s_\beta\}$ of disjoint segments; 
(b) set $S = \{s_\alpha,s_\beta\}$ of intersecting segments}
\label{fig:base-case}
\end{figure}

Given an arc $\alpha\in G$ and a  point 
 $x \not\in s_\alpha$, 
let  $r(x,s_\alpha)$  denote the ray emanating from $x$ in 
the direction $\overrightarrow{px}$, where $p$ 
is the point in $s_\alpha$ closest to $x$ (see Fig.~\ref{fig:fregs}).

\begin{figure}[h]
\centering
\includegraphics[page = 1]{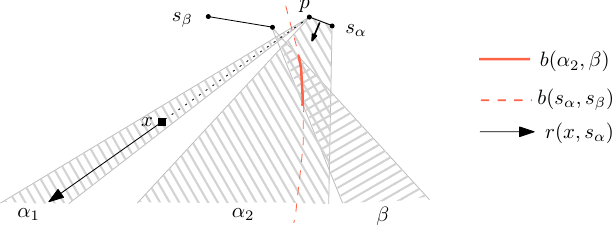} 
\caption{Arcs $\alpha_1, \alpha_2$ of segment $s_\alpha$, arc $\beta$ of segment $s_\beta$; 
attainable regions $R(\alpha_1), R(\alpha_2), R(\beta)$ (shaded); segment bisector $b(s_\alpha,s_\beta)$ (red, dashed) and 
the arc bisector $b(\alpha_2,\beta)$ (red, bold)}
\label{fig:fregs}
\end{figure}

\begin{definition}
\label{def:attainable}
A point $x$, $x \not\in s_\alpha$,  is said to be attainable from $\alpha$  if the direction of
ray $r(x,s_\alpha)$ is contained in $\alpha$. 
An endpoint of segment $s_\alpha$ is attainable from all its corresponding single-vertex arcs. 
All points in $s_\alpha$ are attainable from $\alpha$ if $\alpha$ is a segment arc. 
The locus of points attainable from arc $\alpha$ is called the
\emph{attainable region} of $\alpha$, $R(\alpha)$. See Fig.~\ref{fig:fregs}. 
\end{definition}

The attainable region $R(\alpha)$ is delimited by  two rays  in the direction of the arc-endpoints of $\alpha$ that emanate from the
relevant endpoint(s) of $s_\alpha$, see Fig.~\ref{fig:fregs}.
Within its attainable region, an arc $\alpha$ corresponds to a portion of segment 
$s_\alpha$.
In particular, 
if $\alpha$ is a single-vertex
arc, it corresponds to an endpoint of $s_\alpha$ (see $\alpha_1$ and $\beta$ in
Fig.~\ref{fig:fregs});  
if  $\alpha$ is a segment
arc, it corresponds to one
side of $s_\alpha$ including both its endpoints (see $\alpha_2$ in
Fig.~\ref{fig:fregs}).
We restrict a segment arc $\alpha$ to include exactly one hull direction
of $s_\alpha$. If an arc $\beta$ contains both such directions we
split it in two consecutive segment arcs, $\beta_1$ and $\beta_2$,  by
the direction $\nu\in\beta$ co-linear to $s_\beta$.
Because of this convention, the attainable
region of a segment arc $\alpha$ borders exactly one side of $s_\alpha$.

\begin{remark}
\label{rem:att-disj}
For any 
two arcs $\alpha_1, \alpha_2 \in G$ of the
same segment $s_\alpha$, the interiors of their attainable regions are disjoint.
\end{remark}

\begin{definition}
\label{def:distance}
The distance between an arc $\alpha$ and a point
$x \in  \mathbb{R}^2$ is defined as follows: 
\begin{align*}  d(x, \alpha) &= d(x, s_{\alpha}),  \text{ if $x \in R(\alpha)$;} \\
    d(x, \alpha) &= -\infty,  \text{     if $x \notin R(\alpha)$.} 
\end{align*}
\end{definition}

\begin{definition}
\label{def:arcbisector}
The \emph{arc bisector} between two arcs in $G$  is the locus
of points that are attainable and equidistant from both arcs.  
\end{definition}

The following characterization of arc bisectors directly follows 
from their definition.

\begin{lemma}
\label{rem:arc-bisector}
For two arcs  $\alpha,\beta$, where $s_\alpha\neq s_\beta$,
their \emph{arc bisector}
$b(\alpha,\beta)=b(s_\alpha,s_\beta) \cap R(\alpha)\cap R(\beta)$. 
If  $s_\alpha=s_\beta$, then the arc bisector  $b(\alpha,\beta)=
R(\alpha)\cap R(\beta)$, i.e., $b(\alpha,\beta)$ is
the common boundary of $R(\alpha)$ and $R(\beta)$, which is possibly $\emptyset$.
\end{lemma}

If  $s_\alpha=s_\beta$, the arc bisector $b(\alpha,\beta)$ is an
\emph{artificial bisector}. 
If $\alpha$ and $\beta$ are consecutive, 
then $b(\alpha,\beta)$
depends on what is
considered to be the common
endpoint of $\alpha$ and $\beta$, which is a direction $\nu$. 
In this case $b(\alpha,\beta)$  contains  a ray $r$,
in the direction $\nu$, emanating from the endpoint of $s_\alpha$
that is relevant to $\alpha$ and $\beta$.
In addition to $r$, $b(\alpha,\beta)$ may also 
contain  segment $s_\alpha$, if $s_\alpha \in R(\alpha)\cap R(\beta)$. 
The direction $\nu$ is typically derived in this paper by the
hull direction of the segment arc, whose removal made $\alpha$
and $\beta$ consecutive in $G$ (for details see
Sec.~\ref{sec:del-ins}).
Fig.\ref{fig:base-fvd}a illustrates such an artificial bisector between
 arcs
$\alpha$ and $\gamma$, where $s_\alpha=s_\gamma$.
If $\alpha$ and $\beta$ were made consecutive, not because of removing
an arc between them, but because of our convention that segment arcs
contain only one hull direction, 
then  $\nu$ is  the direction of a ray through $s_\alpha$, in the
maximal arc $\alpha\beta$.

\begin{definition}
\label{def:freg}
The farthest Voronoi region of an arc $\alpha$ in an arc sequence $G$ is
\[ \textit{freg}(\alpha)
= \{x \in \mathbb{R}^{2} \mid d(x, \alpha) > 
d(x, \gamma), \forall \textrm{ arc } \gamma \in G, \gamma\neq \alpha
\}. \]

\noindent
If the boundary of $\freg(\alpha)$ consists solely of arc bisectors,
then 
$\freg(\alpha)$ is called \emph{proper}.
Let $\fregc(\alpha)$ denote the closure of $\freg(\alpha)$.
\end{definition}

\begin{definition}
The partition of the plane defined by the farthest Voronoi  regions  of the
arcs in $G$ and their boundaries,  
is called  the
\emph{farthest Voronoi diagram} of $G$, and is denoted $\fvd(G)$. 
In particular, 

\[\T(G)=\mathbb{R}^{2} \setminus \displaystyle{\cup_{\alpha\in
    G}}\freg(\alpha). \]

\noindent
If $\freg(\alpha), \forall \alpha\in G$, is proper, then  $\T(G)$ as
well as the arc sequence $G$ are
also called proper. 
\end{definition}

Fig.~\ref{fig:base-fvd}a illustrates $\fvd(G)$ for an arc sequence
$G=\alpha\gamma\delta$ where $s_\alpha=s_\gamma$, where  $G$ is a subsequence
of $\gm(S)$ in  Fig. \ref{fig:base-case}b.
Fig.~\ref{fig:base-fvd}b illustrates $\fvd(\Gs)$ ($\Gs$ is 
the simplified version of $G$);  $\hat{\alpha\gamma}$ denotes a
single maximal arc for segment $s_\alpha$.
In Fig.~\ref{fig:base-fvd}a,  note 
that
arcs $\alpha$ and $\gamma$ are separated by $\nu(\beta)=\nu(s_\beta)$. 
The common boundary of $\freg(\alpha)$ and $\freg(\gamma)$ is a portion of the artificial 
bisector $b(\alpha,\gamma)$; 
the whole bisector $b(\alpha,\gamma)$ is the union of this portion and the segment $s_\alpha$.
The attainable regions of $\alpha$ and $\beta$ in Fig.~\ref{fig:base-fvd}a are illustrated by different shading (respectively, falling and rising).
In Fig. \ref{fig:base-case}b, shading indicates $\freg(\hat{\alpha\gamma})$ in $\fvd(\Gs)$.

\begin{figure}
\begin{minipage}{0.49\linewidth}
\centering
\includegraphics[page=1]{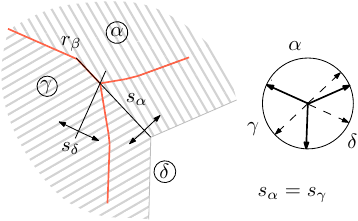}\\
(a)
\end{minipage}
\begin{minipage}{0.49\linewidth}
\centering
\includegraphics[page=2]{base-case-3-c}\\
(b)
\end{minipage}
\caption{(a) $\fvd(G),  G=\alpha\gamma\delta$; (b) $\fvd(\Gs)$ }
\label{fig:base-fvd}
\end{figure}

\begin{lemma}
\label{lemma:freg-disjoint}
For any two arcs $\alpha,\beta \in G$, $\freg(\alpha)\cap \freg(\beta)=\emptyset$.
\end{lemma}
\begin{proof}
By the definition of $\freg(\alpha)$,  any point in $\freg(\alpha)$
must be attainable from $\alpha$, thus,  $\freg(\alpha)\subseteq R(\alpha)$.
By Remark \ref{rem:att-disj}, if $s_\alpha=s_\beta$ then the
interiors of $R(\alpha)$ and $R(\beta)$ are disjoint; thus,
$\freg(\alpha)\cap \freg(\beta)=\emptyset$.
Suppose $s_\alpha\neq s_\beta$. Then if there is a point $x\in
\freg(\alpha)\cap \freg(\beta)$, it must be attainable from both $\alpha$
and $\beta$. But then $x$ can not be in both $\freg(\alpha)$ and
$\freg(\beta)$ (by the definition). Thus, $\freg(\alpha)\cap
\freg(\beta)=\emptyset$.
\qed
\end{proof}

\begin{lemma}
\label{lemma:cover}
If $\partial\freg(\alpha), \forall \alpha \in G$, consists
solely of portions of arc bisectors involving $\alpha$, then
 $\displaystyle{\cup_{\alpha\in G}}\fregc(\alpha)=\mathbb{R}^{2}$.
\end{lemma}

\begin{figure}
\centering
\includegraphics{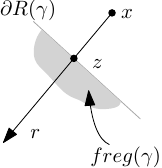}
\caption{Illustrations for the proof of Lemma~\ref{lemma:cover}.}
\label{fig:lemma-cover}
\end{figure}

\begin{proof}
Suppose for the sake of contradiction, that there is a point $x \in \mathbb{R}^2$, 
such that $x \notin \fregc(\alpha)$ for any $\alpha \in G$.
Then $x$  can not be attainable from any arc in $G$ (by the definition
of  $\freg(\alpha)$)
because if it did then $x \in \fregc(\alpha)$ for some $\alpha$ that
$x$ is attainable from.
Consider a ray $r$ originating at $x$ in a direction $\phi$. 
See Fig.~\ref{fig:lemma-cover}.  
Points on $r$ far enough from $x$ must be  attainable from an arc $\alpha \in G$ such that  $\phi \in \alpha$.
Thus, when moving on $r$ starting at $x$,  there is a  point $z$ at that we must 
cross the boundary of an attainable region for
the first time.
 This is the attainable region  of some arc $\gamma \in G$ (not necessarily $\gamma = \alpha$). 
Thus, $d(z,\gamma) \neq -\infty$ and $d(z, \beta) = -\infty$ for any arc $\beta \in G, \beta \neq \gamma$. 
Since the 
interior of segment $xz$
is outside  $R(\gamma)$, $z$ lies on the boundary of $\freg(\gamma)$
(shown shaded in Fig.~\ref{fig:lemma-cover}). 
However, $z$ is not attainable from any arc other than $\gamma$, thus,
$z$ is not on any arc bisector; a contradiction. \qed
\end{proof}

\begin{theorem} \label{theorem:fvd-is-tree}
  If all edges of $\T(G)$ are pieces of arc bisectors ($G$ is a proper arc
  sequence) then  $\T(G)$ is a tree. 
\end{theorem}

\deleted{
\begin{lemma} \label{lemma:fvd-is-tree}
For a  proper arc sequence $G$, $\T(G)$ is a tree. 
\end{lemma}
}

\begin{proof}
Since $G$ is proper, $\T(G)$ is a plane graph and 
each of its edges is a portion of an arc bisector. 
Suppose first 
that $G$ has been simplified ($G=G^s$).
Let $x$ be a point in $\freg(\alpha)$ for some $\alpha \in G$.
We first prove that the entire ray $r(x,s_\alpha)$ 
is enclosed in $\freg(\alpha)$, thus, $\freg(\alpha)$ is unbounded.

Consider any arc $\gamma \in G$, other than $\alpha$, that is attainable from $x$. 
Since  $d(x,\alpha)= d(x,s_\alpha)> d(x,\gamma)= d(x,s_\gamma)$, 
arc bisector $b(\alpha,\gamma)$, which is a portion of
$b(s_\alpha,s_\gamma)$, cannot intersect $r(x,s_\alpha)$.
This is easy to see by considering a disk $D(y)$
centered at a point $y$ of $r(x,s_\alpha)$, see Fig.~\ref{fig:lemma-tree}a.  
As $y$ moves along
$r(x,s_\alpha)$, $D(y)$ enlarges and it must always
intersect $s_\gamma$ (see \cite[Lemma~1]{AuDrysdale06}). 
Thus, no arc bisector $b(\alpha,\gamma)$, where $\gamma$ is attainable
from $x$, 
can 
intersect $r(x,s_\alpha)$ as we walk on $r(x,s_\alpha)$ towards infinity starting at $x$. 
Suppose now that an arc $\delta$ that is not attainable from $x$, 
suddenly becomes attainable as we walk along $r(x,s_\alpha)$, because 
$r(x,s_\alpha)$ intersects $R(\delta)$ at a point $z$.
If $d(z,\delta)<d(z,\alpha)$ then $z$ and a neighborhood around it
must remain in $\freg(\alpha)$. 
Since $z$ is attainable from $\delta$, arc bisector
$b(\alpha,\delta)$ cannot intersect $r(z,s_\alpha)$ for the same
reason as above.
If on the other hand, $d(z,\delta)\geq d(z,\alpha)$ (see Fig.~\ref{fig:lemma-tree}b), 
then $z\in \T(G)$
without being 
on an  arc bisector, contradicting 
our assumption about $\T(G)$.
Thus, no 
point of an arc bisector involving $\alpha$  can be on $r(x,s_\alpha)$.
Since  $\T(G)$ contains only  
arc bisectors, it follows that  
the entire ray $r(x,s_\alpha) \subset \freg(\alpha)$.

\begin{figure}

\centering
\begin{minipage}{0.32\linewidth}
\centering 
\includegraphics[page=1]{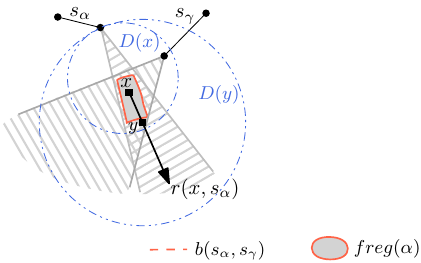}
\\
(a)
\end{minipage}
\hfill
\begin{minipage}{0.32\linewidth}
\centering
\includegraphics[page=2]{lemma-tree-1}
\\
(b)
\end{minipage}
\begin{minipage}{0.32\linewidth}
\centering
\includegraphics[page=3]{lemma-tree-1}
\\
(c)
\end{minipage}
\caption{Illustration for the proof of Theorem~\ref{theorem:fvd-is-tree}
}
\label{fig:lemma-tree}
\end{figure}

It remains to show that $\T(G)$ is connected. Suppose otherwise.  
Then there is a face in $\fvd(G)$,  which is unbounded along two
directions that belong to two  different non-consecutive arcs of $G$.
That is, there are non-consecutive arcs
$\alpha_1,\alpha_2$, 
such that  $\freg(\alpha_1)$ and $\freg(\alpha_2)$ are connected by a
path (see Fig.~\ref{fig:lemma-tree}c).
But then a point $x$ would exist on this path such that $x \in
\freg(\alpha_1)\cap \freg(\alpha_2)$, contradicting
Lemma~\ref{lemma:freg-disjoint}.

If $G$ is not simplified, i.e., $G\neq \Gs$,
then $\T(G)$ can be  derived from $\T(\Gs)$  
by adding the portions of artificial bisectors
within the regions of maximal arcs.
These artificial bisectors involve consecutive arcs of the same
segment,
thus, by their definition, they are half-infinite polygonal lines;
they may be a single ray emanating from
the segment endpoint or they may  actually contain the segment
itself, followed by the ray.
By their definition, these rays cannot intersect transversally, nor can 
they intersect the region boundary, other than the Voronoi vertex on
the boundary that they emanate from, see e.g. Fig.~\ref{fig:base-fvd} where one such artificial bisector is shown.
Thus, after adding artificial bisectors within each region of a
maximal arc, the structure of  $\fvd(G)$ remains a tree, assuming that 
 $\fvd(\Gs)$  is a tree.
\qed
\end{proof}

We remark that for an arbitrary
subsequence $G$ of $\gm(S)$, $\T(G)$ need not necessarily be proper.
The attainable regions of $G$ need not even cover the
plane. In this case  $\T(G)$ 
may contain a two-dimensional region as shown in Fig.~\ref{fig:bad-cases}(a).
However, the intermediate diagrams for arc sequences produced by our
algorithms are always proper.

We also remark that in a proper $\T(G)$, an edge separating
regions of the same segment may coinside with a portion of the 
segment itself, see e.g., Fig.~\ref{fig:bad-cases}(b). 
This edge is a  portion of an (artificial)  arc bisector between  arcs of the same segment. 

\begin{figure}
\centering 
\includegraphics{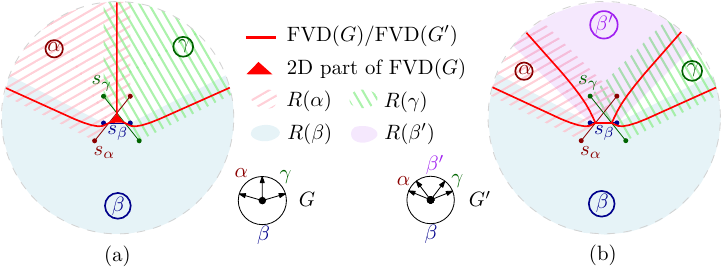}
\caption{(a) A non-proper Voronoi diagram; attainable regions do not cover the plane. (b) A proper Voronoi diagram containing an edge, which coinsides with a portion of the segment $s_\beta$.}
\label{fig:bad-cases}
\end{figure}

\section{A deletion and insertion operation in a sequence of arcs}
\label{sec:del-ins}

Throughout our algorithms we use a deletion and an insertion operation for
original arcs in sequences derived from $\gm(S)$.
The deletion operation produces subsequences of $\gm(S)$; the
insertion operation creates \emph{proper augmented 
subsequences} of $\gm(S)$ and introduces \emph{new} arcs.
The insertion operation for an original arc $\beta$ of $\gm(S)$ corresponds to inserting
$\freg(\beta)$ in the Voronoi diagram computed so far.

 \subsection{Arc deletion}
 \label{sec:deletion}
Let $G$ be a subsequence of $\gm(S)$ and let $\beta$ be an arc in $G$.
Let $G \ominus \beta$ denote the arc sequence derived from
 $G$  after deleting $\beta$. 

When  deleting $\beta$ from $G$, the neighboring
arcs $\alpha$ and $\gamma$ \emph{expand} over $\beta$, see Fig.~\ref{fig:deletion}. 
Either both $\alpha$ and $\gamma$ expand (see Figs.~\ref{fig:deletion}a,c,d) 
or one expands while the other shrinks (see Fig.~\ref{fig:deletion}b).
During the expansion, $\alpha$ or/and $\gamma$  may change from
being a single-vertex arc to being a segment arc. 
Since $\alpha$ and $\gamma$ are original arcs, both
remain  in $G\ominus \beta$.
Their common endpoint $\nu(\alpha,\gamma)$ is determined as follows:

\begin{enumerate}
\item 
If $s_\alpha \neq s_\gamma$, then $\nu(\alpha,\gamma)$ is the first direction of $b(s_\alpha,s_\gamma)$ encountered as we move on $K$
 from $\alpha^*$ to $\gamma^*$ going over the (deleted) arc
 $\beta$. 
 In Fig.~\ref{fig:deletion}c, 
$\nu(\alpha,\gamma)$ is shown as an unfilled circle;
the three (possible) directions of 
$b(s_\alpha,s_\gamma)$ are shown as square marks. 
\item 
If $s_\alpha = s_\gamma \neq s_\beta$, then  $\beta$ must be a
segment arc. We set $\nu(\alpha,\gamma)$
to $\nu(\beta)$, which is the hull direction of $s_\beta$ within $\beta$,
see Fig.~\ref{fig:deletion}d.
Note that setting $\nu(\alpha,\gamma)$ defines the  artificial arc
bisector $b(\alpha,\gamma)$ to contain a ray in the direction
$\nu(\alpha,\gamma)$. 
\item 
If $s_\alpha = s_\beta = s_\gamma$, let $\nu(\alpha,\gamma)= \nu(\alpha,\beta)$.
\end{enumerate}

\begin{figure}[h]
\centering
\includegraphics[page=2]{deletion-rev} \\
(a)\\
\includegraphics[page=3]{deletion-rev}\\
(b)\\
\includegraphics[page=4]{deletion-rev}\\
(c) \\
\includegraphics[page=1]{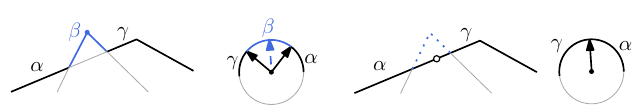} \\
(d)
\caption{The deletion process: $G=\alpha\beta\gamma$ dual and primal image (left); 
the result of deletion of $\beta$ (right)}
\label{fig:deletion}
\end{figure}

\subsection{Arc insertion}
\label{sec:insertion}

Let $G'$ be a proper augmented subsequence of $\gm(S)$ and let $\beta^*$
be a core arc of $\gm(S)$ that does  not appear in $G'$.
Let $G' \oplus \beta^*$ (equiv.  $G' \oplus \beta$) denote 
the arc sequence that results after inserting $\beta^*$ in $G$, and
let $\beta$ denote the expanded version of $\beta^*$ ($\beta\supseteq
\beta^*$ ) as it appears in the resulting subsequence.
Without loss of generality, assume that $\beta^*$ is on the upper
$\gm(S)$, thus, in dual space, $\beta^*$ is on the boundary of the lower
wedge $w$ of $s_\beta$.

The insertion of $\beta^*$ in $G'$ is easy to follow in dual space,
considering a portion of the lower wedge of $s_\beta$.
The endpoints of $\beta$ in dual space are exactly the first intersections of
$w$ with
(the dual of) $G'$, as we walk along $w$ in two
directions, starting at $\beta^*$.
Any portion of $G'$ below $\beta$ is deleted from  $G'$
and is substituted by $\beta$ in $G' \oplus \beta$;
see e.g., the left part of Fig.~\ref{fig:deletion}a,c,d and
Fig.~\ref{fig:reinsert2}.
As Fig.~\ref{fig:deletion} illustrates deletion, $G'$  is shown in the
right part of the figure.
Inserting back $\beta^*$ in $G'$ results
in the original sequence $\alpha\beta\gamma$ in most cases, except
Fig.~\ref{fig:deletion}b.
The insertion of $\beta^*$ in this case is shown in Fig.~\ref{fig:reinsert2}.

\begin{figure}
\centering
\includegraphics{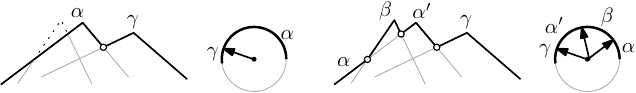}
\caption{Re-insertion of arc $\beta$ for the Fig.~\ref{fig:deletion}b}
\label{fig:reinsert2}
\end{figure}

Let $\alpha,\gamma$ be two consecutive original arcs in $G'$, such
that $\beta^*$ is between $\alpha^*$ and $\gamma^*$ in $\gm(S)$.
Assume a counterclockwise traversal. 
For simplicity, let us first assume that $\alpha$ and $\gamma$
are consecutive in $G'$, however, note that a number of new arcs may
be present between them.
The endpoints of $\beta$, $\nu(\alpha,\beta)$ and
$\nu(\beta,\gamma)$, can be determined by a simple case
analysis as described below. 
They correspond to the two infinite rays bounding $\freg(\beta)$ in
$\fvd(G'\oplus\beta)$.

\paragraph{Inserting $\beta$ between $\alpha,\gamma$
  in $G'$; a case analysis.} 
\begin{enumerate}
\item Suppose $s_\alpha\neq s_\beta, s_\gamma\neq s_\beta$, and $\nu(\alpha,\gamma)\in \beta^*$  (note that $s_\alpha$ might 
or might not equal $s_\gamma$).
Then $\nu(\alpha,\beta)$ is the first direction of
$b(s_\alpha,s_\beta)$ encountered as we walk on
  $K$, starting at $\nu(\alpha,\gamma)$, and moving towards $\alpha^*$. 
  Symmetrically for $\nu(\beta,\gamma)$.
Note that in $G'$, $\nu(\alpha,\beta)\in \alpha$ and $\nu(\beta,\gamma)\in
\gamma$ as $\alpha, \beta, \gamma$ are all original arcs.
See e.g., Fig.~\ref{fig:deletion}a,c,d in the right-to-left direction
and Fig.~\ref{fig:insert}a.

\item Suppose $s_\alpha\neq s_\beta, s_\gamma\neq s_\beta$, but
  $\nu(\alpha,\gamma)\not \in \beta^*$. Then  $\beta^*$ is entirely
  contained in either $\alpha$ or  $\gamma$, see
  Fig.~\ref{fig:reinsert2}. 
Suppose without loss of generality that $\beta^*\subset\alpha$.
 Consider the two consecutive directions,  $\nu_1$,  $\nu_2$, of
 $b(s_\alpha, s_\beta)$ that surround $\beta^*$. 
Since $\beta^* \subset \alpha$, at least one of these directions, say  $\nu_1$, is in $\alpha$;  
let $\nu(\alpha,\beta) = \nu_1$.
\begin{itemize}
\item If $\nu_2 \not\in \alpha$, then 
$\nu(\beta,\gamma)$
is determined as in case 1. 
\item
If $\nu_2 \in \alpha$,
then $\alpha$ is split by $\beta$ into two parts: let $\alpha'$ be the
piece of $\alpha$ neighboring  $\gamma$; let
 $\nu(\alpha',\beta) = \nu_2$ and $\nu(\alpha',\gamma)= \nu(\alpha,\gamma)$.
 Inserting $\beta$ in $G'$ results in substituting  $\alpha\gamma$ by 
 $\alpha\beta\alpha'\gamma$, 
creating a \emph{new} arc $\alpha'$. 
In $\fvd(G'\oplus\beta)$,
$\freg(\beta)$ splits $\freg(\alpha)$ into  two regions:
$\freg(\alpha)$ and  $\freg(\alpha')$.
\end{itemize}

\item  Suppose $s_\alpha=s_\beta$ (symmetrically, if $s_\beta=s_\gamma$). Then 
$\alpha$ is split in two parts by $\nu(\alpha,\beta)$, and the part
 containing $\beta^*$ is called $\beta$. 
Note that $\nu(\alpha,\beta)$ was determined when $\alpha^*$ and
$\beta^*$ became consecutive in some deletion operation.
Note also that 
$\alpha^*,\beta^*$ cannot be neighbors in $\gm(S)$.
In $\fvd(G')$, $\freg(\alpha)$ is simply split in two parts by the artificial
bisector, which is  implied by  $\nu(\alpha,\beta)$: one part remains
$\freg(\alpha)$ and the other part becomes $\freg(\beta)$. 
\end{enumerate}

Fig.~\ref{fig:insert} shows the insertion of an arc $\epsilon$ in the
diagram of Fig.~\ref{fig:base-fvd}a; Fig.~\ref{fig:insert}a
illustrates case~1 and Fig.~\ref{fig:insert}b illustrates case~2.

We now remove the assumption that 
$\alpha$ and $\gamma$ are consecutive in $G'$. 
The procedure to insert $\beta$ is as follows: 
Starting with any point in $\beta^*$, we move counterclockwise on $G'$ until we
encounter an arc $\delta$ with $\nu(s_\delta,s_\beta)$ in $\delta$. 
Note that $\nu(s_\delta,s_\beta)$ is a direction of $b(s_\delta,s_\beta)$.
Note also that $\delta$ may equal $\alpha$.
Clearly, any arcs visited prior to $\delta$ (if any) 
must be deleted. 
Symmetrically, we  move clockwise on $G'$,  until we encounter an arc
$\epsilon$ that contains $\nu(s_\beta,s_\epsilon)$. Since $\beta^* \in \gm(S)$,
for any arc $\omega \in G'$ the directions of $b(s_\beta,s_\omega)\not
\in \beta^*$. 
Arcs $\delta,\epsilon$ play the role of $\alpha,\gamma$ respectively
in  computing $\beta$.

\begin{figure}[h]
\centering
\begin{minipage}{0.49\linewidth}
\centering
\includegraphics[page=3]{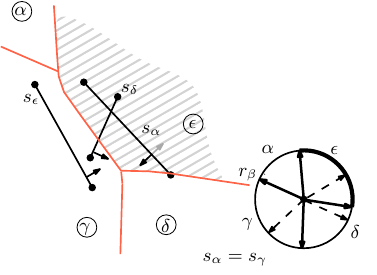} \\
(a)
\end{minipage}
\begin{minipage}{0.49\linewidth}
\centering
\includegraphics[page=4]{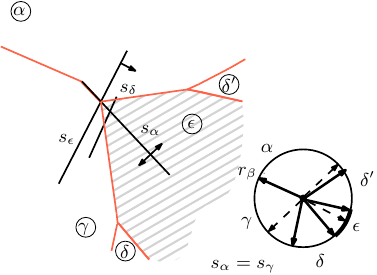} \\
(b)
\end{minipage}

\caption{ 
 The result of inserting arc $\epsilon$ in the diagram of Fig.~\ref{fig:base-fvd}a; 
the insertion corresponds to (a) case  1; (b) case 2. 
}
\label{fig:insert}
\end{figure}

The endpoints of $\beta$ indicate  the startpoints of  $\partial \freg(\beta)$.
Once  they are known,  
$\partial \freg(\beta)$  can be traced  within $\fvd(G')$ in a standard
way.
In particular, we start with one of the startpoints, say
$\nu(\beta,\gamma)$, 
which indicates an unbounded portion  of $b(\beta,\gamma)$,
and trace $\partial \freg(\beta)$ within $\freg(\gamma)$.
As soon as 
$b(\beta,\gamma)$ hits  $\partial \freg(\gamma)$ at a point $x$,  
we start following $b(\beta,\delta)$ in $\freg(\delta)$, where
$\freg(\delta)$ is adjacent to $\freg(\gamma)$ at point $x$.
This process results in tracing a \emph{merge curve}
$m(\beta)$ through $\fvd(G')$, which starts 
at one endpoint of $\beta$ and ends at its second endpoint. 
We denote the traced curve $m(\beta)$; note that it is either the whole curve traced at once, or a concatenation of the two branches meeting at the same point on $s_\beta$.  
We assign $\freg(\beta)$ to be the part of $\fvd(G')$ that is enclosed by
$m(\beta)$ and is unbounded in the directions of $\beta$.
Below we prove correctness of arc insertion: Lemma~\ref{lemma:proper-insert} shows that $m(\beta)$ is \emph{proper}, 
and Lemma~\ref{lemma:correct-insert} shows that the region computed as a result of the insertion operation is 
indeed $\freg(\beta)$ in $\fvd(G'\oplus \beta)$.

\begin{lemma}
\label{lemma:proper-insert}
If $\fvd(G')$ is proper, then the merge curve $m(\beta)$ 
is 
a  simple
connected curve, consisting solely of arc
bisectors involving $\beta$, i.e., $m(\beta)$ is proper.  
\end{lemma}

\begin{proof}
Let $\alpha$ and $\gamma$ be the two neighbors of $\beta$ in
$G'\oplus\beta$; they may be original or new arcs.

If the simplified sequence $G'_s$ is a single maximal arc (i.e.,
$\fvd(G'_s) =\emptyset$ and $s_\alpha=s_\gamma$)
then $m(\beta)$ is a
single branch  of $b(s_\alpha,s_\beta)$.
If $s_\beta = s_\alpha$ or $s_\beta = s_\gamma$ then $G'_s=
(G'\oplus\beta)_s$, thus,
$m(\beta)$ consists of a portion of an artificial arc
bisector ($b(\alpha,\beta)$ or $b(\beta,\gamma)$)
and a portion of a region  boundary in $\fvd(G')$ 
(resp. $\partial \freg(\alpha)$ or $\partial\freg(\gamma)$). Thus, 
$m(\beta)$ is proper.

\begin{figure}
\begin{minipage}{0.49\linewidth}
\centering
\includegraphics[page=1, width=\textwidth]{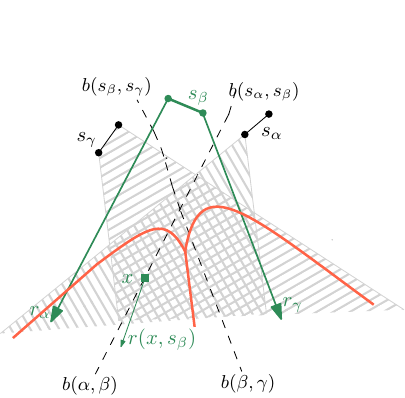}
\\
(a)
\end{minipage}
\hfill
\begin{minipage}{0.49\linewidth}
\centering
\includegraphics[page=2,width=\textwidth]{insert-proof-1}
\\
(b)
\end{minipage}
\caption{Illustrations for the proof of Lemma~\ref{lemma:proper-insert}: 
insertion of arc $\beta$ in $\fvd(G')$ between arcs $\alpha$ and $\gamma$. 
Three segments $s_\alpha, s_\beta, s_\gamma$; 
rays $r_\alpha, r_\gamma$; bisectors $b(\alpha,\beta)$ and $b(\beta,\gamma)$ in dashed lines. 
(a) $R(\alpha)$ and $R(\gamma)$
in $\fvd(G')$ marked by rising and falling tiling pattern; 
$\partial \freg(\alpha)$ and $\partial \freg(\gamma)$ in $\fvd(G')$ in red lines. 
(b) $m(\beta)$ (thick purple lines), the deleted part of $\fvd(G')$ shown in dotted lines. 
}
\label{fig:insert-proof}
\end{figure}

From now on we assume that $s_\alpha\neq s_\beta$, $s_\gamma\neq s_\beta$,  and $\fvd(G'_s)\neq\emptyset$.
Consider the attainable region $R(\beta)$ and let $r_\alpha$,
$r_\gamma$ denote the two open rays bounding it, one in the direction of $\nu(\alpha,\beta)$,
and one in
the direction of $\nu(\beta,\gamma)$, respectively, see Fig.~\ref{fig:insert-proof}. 
We will show that $m(\beta)$ can never hit $r_\alpha$ or $r_\gamma$,
thus, $m(\beta)\subseteq R(\beta)$, and that the vertices of $m(\beta)$ are
proper defined by arc bisectors involving $\beta$.

We perform a bi-directional tracing of $m(\beta)$ starting at
$b(\alpha,\beta)$ and at $b(\beta,\gamma)$.
Consider some intermediate step of the bi-directional tracing. There are two chains of  $m(\beta)$
traced so far: one starting from $b(\alpha,\beta)$ and one starting
from $b(\beta,\gamma)$; we call them  respectively the \emph{left} and the
\emph{right} chain and orient them as we walk away
from their starting bisectors. Suppose these chains are proper, i.e., their edges are portions of arc
bisectors involving $\beta$ and their vertices  are proper, but the
endpoints of their last edges, $e_\ell$  and $e_r$ 
respectively, have not been determined yet.
Edge $e_\ell$ 
cannot hit $r_\alpha$  and $e_r$ cannot hit $r_\gamma$, because  of the visibility property of a segment bisector:
for each point $x$ on the
bisector, the entire ray $r(x,s_\beta)$ is in one side of the
bisector, the side opposite to the one  containing $s_\beta$.
Thus, $r(x,s_\beta)$
must be entirely to the right side of the left chain of $m(\beta)$.

In addition,
neither $e_\ell$ nor $e_r$ can extend to infinity (towards their orientation) within $R(\beta)$.
In fact no arc bisector involving $\beta$, other than
$b(\beta,\alpha), b(\beta,\gamma)$,
can extend to infinity
within $R(\beta)$ because the directions included in $\beta$, other
than its endpoints,  are
distinct from any other arc.

Let $\delta$ and $\phi$ be the arcs in $G'\oplus \beta$ that together
with $\beta$ define respectively $e_\ell$  and $e_r$.
We show that at least  one of the
following holds:
(1) $e_\ell$ hits $\partial\freg(\delta)$ 
before it hits 
 $r_\gamma$; or
(2) $e_r$ hits $\partial\freg(\phi)$ 
before it hits 
$r_\alpha$.
Suppose (1) does not hold, i.e.,  $e_\ell$ hits $r_\gamma$ first. 
Then  the entire (open) ray $r_\alpha$ is outside $\freg(\phi)$, 
and $e_r$ must exit $\freg(\phi)$, i.e., it must hit
$\partial\freg(\phi)$, before it can hit $r_\alpha$, thus, (2) holds. 
Symmetrically, if (2) does not hold, then (1) must hold. 
Since $\partial\freg(\delta)$   and $\partial\freg(\phi)$ are proper,
the intersection  point of either (1) or (2) is the intersection of two arc
bisectors, and thus, a valid vertex.
Note that $m(\beta)$ might hit $s_\beta$; the point where it happens
is on an artificial arc bisector involving another arc of $s_\beta$.

By tracing $m(\beta)$ this way, we follow two
chains within $R(\beta)$ formed by portions of arc bisectors involving $\beta$ separated by Voronoi vertices.  
By the visibility properties of the arc bisectors, 
no Voronoi vertex 
can be visited twice and there is finite number of vertices; moreover neither of the two chains can 
end at infinity within $\beta$. Therefore, the two chains meet, making a single simple path in the arrangement. 
This completes the argument.
\qed
\end{proof}

By Lemma~\ref{lemma:proper-insert}, $m(\beta)$ subdivides the plane into two regions; one of which 
is unbounded in the directions that correspond to arc $\beta$ in $G'\oplus \beta$.
Let us denote the interior of this region $reg(\beta)$. 
The following lemma completes the proof of correctness of the arc insertion procedure. 

\begin{lemma}
\label{lemma:correct-insert}
If $\fvd(G')$ is proper,  $reg(\beta)$ equals $\freg(\beta)$ in
$\fvd(G'\oplus \beta)$. 
\end{lemma}

\begin{proof}
 By  construction of $m(\beta)$ and by the properties of arc bisectors, for any point $x\in reg(\beta)$ we
 have $d(x,\beta) > d(x,\delta)$, where $\delta$ is any  arc in $G'\oplus \beta$ except $\beta$. 
Now suppose there is a point $x \in \freg(\beta)$ in  $\fvd(G'\oplus \beta)$, such that $x \not\in reg(\beta)$. Then the ray $r(x,s_\beta)$ 
intersects the boundary of $\freg(\beta)$ at some point $y \in b(\beta,\delta)$ for some arc $\delta \in G'$. 
However, every point in $r(x,s_\beta)$ is further from $\beta$
than from any other arc in $G' \oplus \beta$, by the argument from the proof of Theorem~\ref{theorem:fvd-is-tree} (see also Fig.~\ref{fig:lemma-cover}a).
We obtain a contradiction. The claim of the lemma follows.
\qed
\end{proof}

\begin{lemma}
\label{lemma:insertion-time}
The diagram $\fvd(G'\oplus \beta)$ can be computed 
from $\fvd(G')$ in $O(|\partial \freg(\beta)| + |\partial \freg(\omega')| + d(\beta))$ time, if a new arc $\omega'$ is created by the insertion of $\beta$.
If no new arcs are created, $\fvd(G'\oplus \beta)$ can be computed in $O(|\freg(\beta)| + d(\beta))$ time. 
Here  $d(\beta)$ is the number of new arcs that get deleted by the insertion of $\beta$.
\end{lemma}

\begin{proof}
To insert $\beta$ into $\fvd(G')$, we first determine the neighbors of $\beta$ in $G' \oplus \beta$
by tracing $G'$ in both directions starting from $\beta^*$. 
Since any visited arc gets deleted from the arc sequence, and since tracing requires $O(1)$ time per visited arc, 
this step requires overall $O(d(\beta))$ time. 

Now we trace $\partial \freg(\beta)$. 
Suppose first that  no new arcs are created by the insertion of $\beta$. This means, that the
neighbors of $\beta$ in $\fvd(G'\oplus \beta)$ are two different arcs.
Thus there is at least one unbounded edge of $\fvd(G')$ that gets deleted  by the insertion of 
$\beta$ (e.g., any of the unbounded edges between the regions of the neighbors of $\beta$). 
We trace $\partial \freg(\beta)$ starting  from this edge in  the standard way,  
see e.g.~\cite{CGbook}. The time complexity of such tracing is proportional to $|\partial\freg(\beta)|$ 
plus the total complexity of the Voronoi regions of the new
arcs that get deleted. Such edges 
form a forest of total complexity $O(d(\beta))$. Thus the tracing requires $O(|\partial\freg(\beta)| + d(\beta))$ time. 

Suppose now that the insertion of $\beta$ caused the creation of a new arc $\omega'$. In this case, no unbounded edge of $\fvd(G')$ is deleted, 
and we trace a number of edges on $\partial \freg(\omega')$ in order to find a starting 
point on $\partial\freg(\beta)$. Thus the tracing in this case requires 
$O(|\partial \freg(\beta)| + |\partial \freg(\omega')|)$ time. \qed
\end{proof}

\section{A Randomized Linear Construction}
\label{sec:randomized}

We present 
an expected linear-time algorithm to compute
$\fvd(S)$, given $\gm(S)$, inspired by the simple two-phase
randomized approach for points in convex
position~\cite{Chew90}, using the structures of Secs.~\ref{sec:fvd}
and~\ref{sec:del-ins}. 
Let $\alpha_1,\alpha_2,\ldots, \alpha_h$ be a random permutation of
 arcs in $\gm(S)$, and 
let $A_i=\{\alpha_1,\alpha_2,\ldots, \alpha_i\}$, $1\leq i\leq h$, be
the  set of the first $i$ arcs in this permutation. 
Let $t$ be the largest index  such that $A_t$
consists of arcs of only two segments, which form exactly two maximal
arcs in the corresponding subsequence of $\gm(S)$.

The algorithm proceeds in two phases as follows.

\paragraph{Phase~1 (deletion phase).}
We  compute the subsequence
$G_i$, $t \leq i < h$, where  $G_h = \gm(S)$, and 
$G_i$ is obtained from
$G_{i+1}$ by deleting arc $\alpha_{i+1}$ as described in
Sec.~\ref{sec:del-ins} ($G_i = G_{i+1} \ominus \alpha_{i+1}$).
The two neighbors of $\alpha_{i+1}$ in $G_{i+1}$ are recorded as a
tentative re-entry point for $\alpha_{i+1}$ at Phase~2.
At the end of Phase~1, we obtain $G_t$, whose simplified version
consists of exactly two maximal arcs.

\paragraph{{Phase~2 (insertion phase)}.} 
We incrementally compute $G'_i$  and $\fvd(G_i')$, for $t < i \leq h$, starting with 
$\fvd(G_{t})$. 
$G_t$ consists only of two maximal arcs, thus, $\fvd(G_{t})$ is proper
and can be computed in time $O(t)$: compute first $\fvd(G_{t}^s)$;
then split its regions by artificial
bisectors as indicated by the endpoints of individual arcs (see e.g.,
Fig.~\ref{fig:base-fvd}).
In addition, enforce our convention on segment arcs:  
if an arc of $G_{t}$ contains both hull directions of its segment, split
it in two consecutive arcs by the direction of a ray through the
segment, and split its region by an artificial arc bisector.
$\fvd(G_{i+1}'), i \geq t,$ is derived from $\fvd(G_i')$ by inserting arc
$\alpha_{i+1}$ ($G'_{i+1}=G'_i\oplus\alpha_{i+1}$) as detailed in
Sec.~\ref{sec:del-ins}.  
At the end of Phase~2, we obtain $\fvd(G_h)=\fvd(\gm(S))$. 

\begin{lemma}
\label{lemma:two-new-primes}
The number of arcs in $G_i'$ is at most 
$2i$  and  the complexity of
$\fvd(G_i')$ is $O(i)$. 
\end{lemma}
\begin{proof}
By Lemma~\ref{lemma:proper-insert}, at each step of Phase 2, 
we are creating at most one new arc.  
Thus, the total number of arcs in $G_i'$ ($|G_i'|$) is at most $2i$.
Since $G_t$ is a base subsequence of $\gm(S)$, $\fvd(G_t)$ is  proper. 
Applying  
Lemmas~\ref{lemma:proper-insert} and~\ref{lemma:correct-insert} for each step of Phase~2, we obtain that 
$\fvd(G_i')$ must also be proper.
By Theorem~\ref{theorem:fvd-is-tree}, $\T(G_i')$ is a tree, with vertices of
degree at least three. Since $\fvd(G_i')$ has exactly one face per arc of
$G_i'$ and $|G_i|\leq 2i$, the claim follows from the Euler's formula.
\qed
\end{proof}

\begin{lemma}
\label{lemma:exp-const-primes}
The expected number of new arcs traced at any step of Phase~2 is
$O(1)$. 
\end{lemma}

\begin{proof}
By construction, $G_i'$ is an augmented subsequence of $\gm(S)$ corresponding to $G_i$. 
Thus, by Lemma~\ref{lemma:two-new-primes}, $G_i'$ consists of
at most $i$ new arcs plus the $i$ original arcs of $G_i$.
To insert arc $\alpha_{i+1}=\beta$ at one step of Phase~2,  the pair $\alpha,\gamma$ of  consecutive
original arcs (that has been stored with $\alpha_{i+1}$) is
picked, and some new arcs between $\alpha$ and $\gamma$ in $G_i'$ may
be traced. 
Since every element of $A_{i+1}$ is equally likely to be 
$\alpha_{i+1}$, each pair of consecutive original arcs in $G'_i$
has probability $1/i$ to be considered at step $i$.
Let $n_j$ be the number of new arcs in-between the $j$th pair of
original arcs in $G_i'$, $1\leq j\leq i$; $\sum_{j=1}^in_j \leq i$.
The expected number of new arcs that are traced is then 
$\sum_{j=1}^i({n_j}/i) \leq 1$.
\qed
\end{proof}

\begin{theorem}
\label{thm:linear-rand}
Given $\gm(S)$, the $\fvd(S)$ can be computed in expected $O(h)$ time, where $h$ is the complexity of $\fvd(S)$. 
\end{theorem}

\begin{proof}
We use backwards analysis, going from $\fvd(G_{i+1}')$ to
$\fvd(G_{i}')$.
Recall that inserting  
$\alpha_{i+1}$ in $\fvd(G'_i)$ takes time proportional
to the complexity of its boundary,  
plus the number of new arcs that get deleted (if any) by the
insertion of $\alpha_{i+1}$,
plus occasionally, the
complexity of its  neighboring Voronoi region.
The latter two additions to the complexity represent a  difference  from the corresponding
argument in the case of points.  
The former addition is expected $O(1)$ due to  Lemma~\ref{lemma:exp-const-primes}; the latter addition is as well  expected $O(1)$
by the following argument. 
Since $A_h$ is a random permutation of the arcs in $G_h$, 
the expected time complexity of inserting $\alpha_{i+1}$ in
$\fvd(G_i')$ is equivalent to the expected complexity of a randomly selected face
in $\fvd(G_{i+1}')$, plus the expected complexity of its immediate neighbor.
Since $\fvd(G_i')$ has size $O(i)$ and it consists of $O(i)$ faces,
the expected complexity of a randomly selected region is constant.
The same holds for one neighbor of the randomly selected region.
Thus, the expected time spent to insert   $\freg(\alpha_{i+1})$ in
  $\fvd(G_i')$ is constant.
Since the total number of arcs is $h$, the claim follows.
\qed
\end{proof}

\section*{Concluding Remarks}
We expect that Theorem~\ref{thm:linear-rand} is applicable to updating
a nearest-neighbor segment Voronoi diagram, after the deletion of one
segment, in expected linear time. 
The adaptation, however,  is non-trivial and we plan to address it in
a subsequent paper.
We also expect Theorem~\ref{thm:linear-rand} to be applicable (after
adaptation) to computing the order-$(k{+}1)$
subdivision within an order-$k$ Voronoi region in expected time proportional to the
complexity of the region's boundary. 
A deterministic linear-time construction for the farthest segment Voronoi
diagram, given its Gaussian map, remains an open problem. In future
research we plan to  investigate 
the linear-time framework of Aggarwal et al.~\cite{AGSS89}.

Note that the randomized incremental
construction for 
the farthest abstract Voronoi diagram~\cite{AbstractFarthestVoronoi} is not related to the randomized
linear-time approach presented in this paper.

\subsubsection*{Acknowledgements.} We wish to thank Kolja Junginger
and Ioannis Mantas at the Universit\`a della Svizzera
  italiana for helpful discussions, reading through this manuscript,  and for producing Figure~\ref{fig:bad-cases}.

\bibliographystyle{splncs03}
\bibliography{voronoi}
\end{document}